\documentclass[a4paper]{article}

\usepackage[margin=1in]{geometry}
\usepackage[pdftex,hyperindex,breaklinks]{hyperref}
\usepackage{amsmath,amsthm}
\usepackage{algorithm}
\usepackage[noend]{algpseudocode}
\usepackage{xspace}
\usepackage[english]{babel}
\usepackage[utf8]{inputenc}
\usepackage[bitstream-charter]{mathdesign}

\newtheorem{thm}{Theorem}[section]   
\newtheorem{lem}[thm]{Lemma}
\newtheorem{cor}[thm]{Corollary}
\newtheorem{prop}[thm]{Proposition}
\newtheorem{definition}[thm]{Definition}
\newtheorem{rem}[thm]{Remark}
\newtheorem*{notations}{Notations}
\theoremstyle{remark}
\newtheorem{example}{Example}

\algnewcommand\Input{\item[\textbf{Input:}]}
\algnewcommand\Output{\item[\textbf{Output:}]}
\newcommand{\K}{\ensuremath{\mathbb{K}}\xspace} 
\newcommand{\gO}{O}
\newcommand{\gOt}{\tilde{O}}
\newcommand{\gOeps}{\gO_\epsilon}
\newcommand{\gOteps}{\gOt_\epsilon}
\newcommand{\cfq}{{\mathbb{F}_{\!q}}}
\newcommand{\cfqs}{{\mathbb{F}_{\!q^s}}}
\newcommand{\gz}{\mathbb{Z}}
\newcommand{\gr}{R}
\newcommand\algo\textsc
\newcommand{\verif}{\algo{VerifySP}\xspace}
\newcommand{\ceil}[1]{\ensuremath{\lceil {#1} \rceil}}
\newcommand{\height}[1]{\|#1\|_\infty}
\newcommand\dilat{^{\smash{[\alpha]}}}
\newcommand\dilatj{^{\smash{[\alpha_j]}}}
\DeclareMathOperator{\supp}{supp}
\DeclareMathOperator{\bquo}{quo}

\title{On Exact Division and Divisibility Testing\\ for Sparse Polynomials}

\author{Pascal Giorgi \hfill Bruno Grenet\hfill Armelle Perret du Cray\\
    LIRMM, Univ. Montpellier, CNRS\\
    Montpellier, France\\
    \{pascal.giorgi,bruno.grenet,armelle.perret-du-cray\}@lirmm.fr}

\begin{document}

\maketitle

\begin{abstract}
  No polynomial-time algorithm is known to test whether a sparse polynomial $G$ divides another sparse polynomial $F$.
  While computing the quotient $Q=F \bquo G$ can be done in polynomial time with respect to the sparsities of $F$, $G$
  and $Q$, this is not yet sufficient to get a polynomial-time divisibility test in general.  Indeed, the sparsity of the quotient
  $Q$ can be exponentially larger than the ones of $F$ and $G$.  In the favorable case where the sparsity $\#Q$ of the quotient is
  polynomial, the best known algorithm to compute $Q$ has a non-linear factor $\#G\#Q$ in the complexity, which is not optimal.

  In this work, we are interested in the two aspects of this problem.  First, we propose a new randomized algorithm that computes
  the quotient of two sparse polynomials when the division is exact. Its complexity is quasi-linear in the sparsities of $F$, $G$
  and $Q$. Our approach relies on sparse interpolation and it works over any finite field or the ring of integers.  Then, as a
  step toward faster divisibility testing, we provide a new polynomial-time algorithm when the divisor has a specific shape.
  More precisely, we reduce the problem to finding a polynomial $S$ such that $QS$ is sparse and testing divisibility by $S$ can
  be done in polynomial time.  We identify some structure patterns in the divisor $G$ for which we can efficiently compute such a
  polynomial~$S$.
  \end{abstract}

\section{Introduction}

The existence of quasi-optimal algorithms for most operations on dense polynomials yields a strong base for fast algorithms in
computer algebra~\cite{MCAlgebra} and more generally in computational mathematics.  The situation is different for
algorithms involving sparse polynomials.  Indeed, the sparse representation of a polynomial
$F = \sum_{i=0}^D f_i X^i \in\gr[X]$ is a list of pairs $(e_i, f_{e_i})$ such that each $f_{e_i}$ is nonzero.
Therefore, the size of the sparse representation of $F$ is $\gO(\#F(B+\log D))$ bits, where $B$ and $\#F$ bound
respectively the size of the coefficients and the number of nonzero coefficients of $F$. 
Polynomial-time algorithms for sparse polynomials need to have
a \mbox{(poly-)}logarithmic dependency on the degree.  
On the one hand, several
$\mathsf{NP}$-hardness results rule out the existence of such fast algorithms unless $\mathsf P=\mathsf{NP}$, for instance for \textsc{gcd} computations~\cite{Plaisted84}. On the other hand, polynomial-time algorithms are known for many important
operations such as multiplication, division or sparse interpolation. We refer to Roche's survey~\cite{Roche2018} for a thorough
discussion on their complexity and on the remaining major open problems.

The main difficulty with sparse polynomial operations is the fact that the size of the output does not exclusively depend on the
size of the inputs, contrary to the dense case. For instance, the product of two polynomials $F$ and $G$ has at most $\#F\#G$
nonzero coefficients. But it may have as few as $2$ nonzero coefficients~\cite{giorgi2020:sparsemul}.
The size growth can be even more dramatic for sparse polynomial division. For instance, the quotient of $F=X^D-1$ by $G=X-1$ is
$F/G=\sum_{i=1}^{D-1} X^i$. The output can therefore be exponentially larger than the inputs. Such a growth is a major 
difficulty to design efficient algorithms for Euclidean division since it is hard to predict the sparsity of the quotient and
the remainder, which can range from constant to exponential.

One important line of work with sparse polynomials is to find algorithms with a quasi-optimal bit complexity
$\gOt(T (\log D + \log C))$ where $T$ is the number of nonzero coefficients of the input and output, $D$ the degree and $\log C$ a
bound on the coefficient bitsize. The problem is trivial for addition and subtraction. For multiplication, though many algorithms have been 
proposed~\cite{ArRo15,Roche2011,MonaganPearce2009,vdHLec2013,vdHLec2012,Johnson74,ColeHariharan,nakos2020},
none of them was quasi-optimal in the general case. Only recently, we 
proposed a quasi-optimal algorithm for the
multiplication of sparse polynomials over finite fields of large characteristic or over the integers~\cite{giorgi2020:sparsemul}. We note that similar results have been given in a more recent preprint, assuming some heuristics~\cite{vdH2020}.

These fast output-sensitive multiplication algorithms strongly rely on \emph{sparse polynomial interpolation}. In the
latter problem, a sparse polynomial is implicitly represented by either a \emph{straight-line program} (SLP) or a
\emph{blackbox}. Though efficient output-sensitive algorithms exist in the blackbox model~\cite{KALTOFEN2003365} they are not well
suited for sparse polynomial arithmetic since one probe of the blackbox is assumed to take a constant time while it is not in our
case. Using sparse interpolation algorithms on SLP is not a trivial solution either since no quasi-optimal bit complexity
bound is known despite the remarkable recent
progress~\cite{Huang2019,ArGiRo16,ArGiRo14,ArGiRo13,ArRo14,vdHLec2014,GiRo2011,GaSch09,BenOrTiwari,vdHLec2019}.
The best known result due to Huang~\cite{Huang2019} has bit complexity $\gOt(L(T\log D \log C))$ to interpolate an
SLP of length $L$ representing a $T$-sparse polynomial of degree at most $D$ with coefficient of size $\log C$.  

In this work, we are interested to use fast sparse interpolation to derive a better complexity bound for sparse polynomial
division, in the special case where the division is exact.
As a second goal, we make progress on the very related problem of testing the divisibility of two sparse
polynomials.

\subsection{Previous work}

\paragraph{Euclidean division of sparse polynomials.}
Let $F = GQ+R \in\K[X]$ where $F$ and $G$ are two polynomials with at most $T$ nonzero coefficients ($\#F$, $\#G\leq T$),
$D=\deg(F) > n=\deg G$, and $\deg R < \deg G$.  Computing $Q$ and $R$ through classic Euclidean division requires $\gO(\#G\#Q)$
operations in $\K$. Yet keeping track of the coefficients of the remainder during the computation dominates the cost, due to many
exponent comparisons. The total complexity is $\gO(\#F + \#Q (\#G)^2)$ using sorted lists, or $\gO(\#F+\#Q\#G\log (\#F+\#Q\#G))$
using binary heaps or the geobucket structure~\cite{yant_geobucket_1998}.
Heap technique has been improved to further lower down the size of the heap.
Johnson proposes an algorithm that uses a heap of size $\#Q+1$~\cite{Johnson74}, and Monagan and Pearce provide a variant that
maintains a heap of size $\gO(\#G)$~\cite{MonaganPearce2007}. The best solution to date for sparse polynomial division is to
switch from a quotient heap to a divisor heap whenever the quotient is getting larger than the divisor. The complexity becomes
$\gO(\#F+\#Q\#G\log\min(\#Q, \#G))$~\cite{MonaganPearce2011}.

To the best of our knowledge, no algorithm has been specifically designed for the special case of exact division.

\paragraph{Sparse divisibility testing.}
The problem of sparse divisibility testing is to determine, given two sparse polynomials $F$ and $G$, whether $G$ divides $F$.  It
is an open problem whether this problem admits a polynomial-time algorithm, that is an algorithm that runs in time
$(T\log D)^{\gO(1)}$ where $T$ bounds the number of nonzero terms of the inputs and $D$ their degrees. We note that the division algorithms do
not settle the problem. Indeed, the quotient of two sparse polynomials $F$ and $G$ can be exponentially larger than $F$ and $G$.

The only general complexity result on this problem is due to Grigoriev, Karpinksi and Odlyzko~\cite{GrigorievKarpinksiOdlyzko96}
who show that the problem is in $\textsf{coNP}$ under the Extended Riemann Hypothesis (ERH). Besides, the problem admits
polynomial-time algorithms in the easy cases where $\deg(G)$, $\deg(F)-\deg(G)$ or $\#Q$ are polynomially
bounded~\cite{Roche2018}.  On the other hand, some related problems are $\textsf{coNP}$-hard, such as the divisibility of a
product of sparse polynomials by a sparse polynomial, the computation of the constant coefficient of the quotient or the degree of
the remainder~\cite{Plaisted84}.

\renewcommand{\thefootnote}{\fnsymbol{footnote}}

\subsection{Our contributions}

We focus on the exact division of sparse polynomials. We first provide algorithms whose bit complexity are quasi-linear in the
input and the output sparsities. Our algorithms work over finite fields and the integers, and are randomized.  Over a finite field
$\cfq$ of characteristic larger than the degree $D$ of the inputs, it computes the quotient of two polynomials in
$\gOteps(T\log D\log q)$\footnote{We let $\gOt(f(n)) = f(n)(\log f(n))^{\gO(1)}$ and
  $\gOteps(f(n)) = \gOt(f(n))\log^{\gO(1)}\frac{1}{\epsilon}$.}  bit operations with probability at least $1-\epsilon$, where $T$
bounds the number of nonzero terms of both the inputs and the output.  For smaller characteristic, the complexity bound is
$\gOteps(T\log^2 D(\log D+\log q))$. For polynomials over $\gz$ with coefficients bounded by $C$ in absolute value, our algorithm
performs $\gOteps(T(\log C+\log D\log S)+\log^3 S)$ bit operations where $S$ is the maximum of $D$ and the absolute value of the
coefficients of the result. Our main technique is to adapt sparse polynomial interpolation algorithms to our needs.  Our work
focuses on the univariate case but it can be straightforwardly extended to the multivariate case using (randomized) Kronecker
substitution~\cite{ArRo14}.  We shall mention that the technique behind our exact division generalizes to the sparse interpolation
of SLPs with divisions.

We also provide a polynomial time algorithm for special cases of the sparse polynomial divisibility testing problem when
$\deg(F) = \gO(\deg(G))$.
We prove that if $G$ contains a small chunk of coefficients, with large \emph{gaps} surrounding it, then one can test in
polynomial time whether $G$ divides $F$. More precisely, we require $G$ to be written as $G_0 + X^k G_1 + X^\ell G_2$ with
$\deg(G_1) = (T\log D)^{\gO(1)}$, $k-\deg(G_0) = \Omega(D)$ and $\ell-\deg(X^k G_1)=\Omega(D)$.  This encompasses polynomials of
the form $G_0 + X^kG_1$ or $G_1+X^\ell G_2$. Our technique is to prove that in this situation, even if the quotient $F/G$ may have
an exponential number of nonzero terms, we are able to efficiently compute a multiple of the quotient that is sparse.

\paragraph{Notations.}
Let $F = \sum_{i=1}^T f_i X^{e_i}$. We use $\#F = T$ to denote its sparsity (number of nonzero terms) and $\supp(F) = \{e_1, \dotsc, e_T\}$ for its support. If $F\in\gz[X]$, we use $\height{F} = \max_i|f_i|$ to denote its height. The \emph{reciprocal} of $F$ is the polynomial $F^\star = X^{\deg(F)} F(1/X)$.

\section{Exact division}\label{sec:exactdiv}

Our method to compute the exact quotient $F/G$ of two sparse polynomials $F$ and $G$ relies on \emph{sparse interpolation}
algorithms. These algorithms usually take as input a straight-line program (SLP), or sometimes a blackbox, representing a sparse
polynomial $Q$, together with bounds on the sparsity and the degree of $Q$. The output is the sparse polynomial $Q$ given by the
list of its nonzero monomials.

There are two main families of sparse polynomial interpolation algorithms. The first one, which originates with the work of
Prony~\cite{Prony} and Ben-Or and Tiwari~\cite{BenOrTiwari}, uses evaluations of $Q$ on geometric progressions. The second one,
initiated by Garg and Schost~\cite{GaSch09}, computes the reduction of $Q$ \emph{modulo} $X^p-1$ for some random prime $p$. This
second line of work is more suitable to our case, to obtain the best complexity bounds.

For polynomials over $\cfq$, we rely on the best known sparse interpolation algorithms 
due to Huang~\cite{Huang2019}
when $q > \deg Q$ and to Arnold, Giesbrecht and Roche~\cite{ArGiRo14} otherwise.
These two algorithms 
compute the reduction of $Q$ \emph{modulo} $X^p-1$ for some random prime $p$. This computation is known as
\emph{SLP probing} and it can use dense polynomial arithmetic when $p$ is small enough. The goal is then to reconstruct $Q$ from
$Q_p = Q\bmod X^p-1$. One difficulty comes from exponent recovery since a monomial $cX^e$ of $Q$ is mapped to $cX^{e\bmod p}$ in
$Q_p$.  A second one, called \emph{exponent collision}, is when two distinct exponents $e_1$, $e_2\in\supp(Q)$ are congruent
\emph{modulo} $p$. These collisions create the monomial $(c_1+c_2)X^{e}$ in $Q_p$, from which neither $c_1X^{e_1}$ nor
$c_2X^{e_2}$ can be directly recovered.

The latter difficulty is handled similarly in~\cite{ArGiRo14} and~\cite{Huang2019}. Taking $p$ at random in a
large enough set of prime numbers, one can show that a substantial fraction of the monomials of $Q$ do not collide during the
reduction \emph{modulo} $X^p-1$. Therefore, working with several random primes $p$ allows the full reconstruction. The two
algorithms mainly differ in the way they overcome the first difficulty.

Huang's very natural approach is to consider the derivative $Q'$ of $Q$~\cite{Huang2019}. If the characteristic of $\cfq$ is
larger than the degree of $Q$, a monomial $cX^e$ of $Q$ is mapped to $ceX^{e-1}$ in $Q'$. Then it is mapped to
$ceX^{(e-1)\bmod p}$ in $[Q']_p= Q' \bmod X^p-1$. Given an SLP for $Q$, one can efficiently compute an SLP for $Q'$ using
automatic differentiation~\cite{BuClSh97}. If the monomial $cX^e$ does not collide \emph{modulo} $X^p-1$, 
it can be retrieved from its images in $Q_p$ and $[Q']_p$ using a mere division on the
coefficients.

With smaller characteristic, Huang's idea is no longer working since not all the integer exponents exist in $\cfq$.  Instead,
Arnold, Giesbrecht and Roche work \emph{modulo} several primes $p_i$ and use the Chinese remainder theorem to recover the
exponent. For, they introduce the \emph{diversification} technique to be able to match the corresponding monomials in
$Q \bmod X^{p_i}-1$. Indeed, replacing $Q$ with $Q(\alpha_j X)$ for several randomly chosen $\alpha_j$'s,  the nonzero
coefficients of $Q(\alpha_j X)$ are pairwise distinct with a good probability.

\medskip

The main difficulty to adapt these approaches to the computation of $Q=F/G$ is that the division in $\cfq[X]/(X^p-1)$ is not
well-defined.
In the next section we show that taking $\alpha$ at random in a sufficiently large set is enough for $G(\alpha X)$ and
$(\alpha X)^p-1$ to be coprime.  We shall mention that this technique may be extended to other sparse polynomial interpolation
algorithms. In particular, this includes slightly faster algorithms~\cite{vdHLec2019}, but they rely on  unproven
\emph{heuristics}.

\subsection{Computation of a reduced quotient} \label{sec:reducedquo}

Given $F$, $G\in\cfq[X]$ such that $F = GQ$, our aim is to compute $Q\bmod X^p-1\in\cfq[X]$ for some prime $p$. Let $F_p = F\bmod X^p-1$,
$G_p = G\bmod X^p-1$ and $Q_p = Q\bmod X^p-1$, then
\begin{equation}\label{eq:quomodp}
F_p = G_pQ_p\bmod X^p-1.
\end{equation}
If $\gcd(G_p, X^p-1) = 1$, then $G_p$ is invertible \emph{modulo} $X^p-1$, and $Q_p$ can be computed. 
Otherwise, Equation~\eqref{eq:quomodp} admits several solutions and does not define $Q_p$. The following lemma defines a
probabilistic approach to overcome the latter difficulty.

\begin{lem}\label{lem:coprime}
  Let $A$ and $B\in\cfq[X]$ be two nonzero polynomials with $B(0)\neq 0$, and $\alpha$ randomly chosen in some extension $\cfqs$ of $\cfq$. Then
  $A(\alpha X)$ and $B(X)$ are coprime with probability at least $1-\deg(A)\deg(B)/q^s$.
\end{lem}

\begin{proof}
  Let $\beta$ be a root of $B$ in an algebraic closure $\overline\cfq$ of $\cfq$. Then $\beta$ is a root of $A(\alpha X)$ if and
	only if $A(\alpha\beta) = 0$, that is $\alpha$ is a root of $A(\beta X)$. Since $A(\beta X)$ is nonzero and $\deg(A(\beta X)) = \deg(A)$, there exist at
  most $\deg(A)$ roots of $A(\beta X)$ in $\overline\cfq$. Since $B$ has at most $\deg(B)$ roots in $\overline\cfq$, there are at
  most $\deg(A)\deg(B)$ values of $\alpha$ such that there exists a common root $\beta$ of $A(\alpha X)$ and $B(X)$. Therefore,
  with probability at least $1-\deg(A)\deg(B)/q^s$, $A(\alpha X)$ and $B(X)$ do not share a common root in $\overline\cfq$, that
  is they are coprime.
\end{proof}

\begin{notations}
  For $A\in\cfq[X]$, $\alpha\in\cfqs$ and $p\ge 0$, let 
  $A\dilat$ be the polynomial
  $A(\alpha X)$, $A_p$ be the polynomial $A(X)\bmod X^p-1$ and $A_p\dilat$ be the polynomial
  $A\dilat(X)\bmod X^p-1 = A(\alpha X)\bmod X^p-1$.
\end{notations}

We remark that $A\dilat_p \neq A_p(\alpha X)$. The idea is to apply Lemma~\ref{lem:coprime} to $G$ and $X^p-1$. Instead of
applying Equation~\eqref{eq:quomodp} to $F$ and $G$, we apply it to $F\dilat$ and $G\dilat$ to get
$Q_p\dilat = Q(\alpha X)\bmod X^p-1$. In other words, we compute $Q_p\dilat$ from the equation
$F_p\dilat = G_p\dilat Q_p\dilat\bmod X^p-1$. If $\alpha$ is chosen at random in some extension $\cfqs$ of $\cfq$, $G\dilat$
and $X^p-1$ are coprime with probability at least $1-p\deg(G)/q^s$ and $G_p\dilat$ is invertible \emph{modulo} $X^p-1$ with the
same probability. Since we compute $Q\dilat_p$ for any $\alpha$, we can adapt the algorithm of Arnold, Giesbrecht and
Roche~\cite{ArGiRo14}.

In order to adapt Huang's algorithm~\cite{Huang2019}, we need to compute $Q'(X)\bmod X^p-1$. To this end, we rely on the equality
\begin{equation}\label{eq:quoprimemodp}
    [F']_p-[G']_pQ_p \bmod X^p-1 = G_p[Q']_p\bmod X^p-1
\end{equation}
where $[A']_p$ denotes $A'(X)\bmod X^p-1$ for any $A\in\cfq[X]$. We notice that this equation is similar to
Equation~\eqref{eq:quomodp}. Knowing $Q_p$, the equation defines $[Q']_p$ if and only if $G_p$ is invertible \emph{modulo}
$X^p-1$. This means that if $\alpha$ is chosen at random in $\cfqs$, Equations~\eqref{eq:quomodp} and~\eqref{eq:quoprimemodp}
allow to compute both $Q\dilat_p$ and $[(Q\dilat)']_p$ with probability at least $1-p\deg(G)/q^s$, where $[(Q\dilat)']_p$
is the polynomial $(Q(\alpha X))'\bmod X^p-1$. Next lemmas give the cost of these operations.

\begin{lem}\label{lem:Aalphap}
  Let $A\in\cfq[X]$ of degree $D$, sparsity $T$ and $\alpha\in\cfqs$. Then $A\dilat$ can be computed in
  $\gOt(T\log Ds\log q)$ bit operations, and $A\dilat_p$ in $\gO(T\log D\log\log p+Ts\log q)$ more bit operations.
\end{lem}

\begin{proof}
  Computing $A\dilat = A(\alpha X)$ requires $T$ exponentiations of $\alpha$, that is $\gO(T\log D)$ operations in $\cfqs$, which
  gives a bit complexity of $\gOt(T\log Ds\log q)$. Computing $A\dilat_p$ from $A\dilat$ requires $T$ exponent divisions, 
	that can be performed in $\gO(T\frac{\log D}{\log p})$ divisions on integers of size $\log p$, 
	and $T-1$  coefficient additions.
\end{proof}

\begin{lem}\label{lem:Qalphap}
  Let $F$ and $G\in\cfq[X]$ such that $G$ divides $F$, and let $p\ge 0$ and $\alpha\in\cfqs$ such that $G\dilat$ and $X^p-1$ are
  coprime. If $Q = F/G$, the polynomials $Q_p\dilat$ and $[(Q\dilat)']_p$ can be computed in $\gOt(T\log Ds\log q+ps\log q)$ bit
  operations, where $D = \deg(F)$ and $T$ is a bound on the sparsities of $F$ and $G$.
\end{lem}

\begin{proof}
  To get $Q\dilat_p$, the first step computes $F\dilat_p$ and $G\dilat_p$. Then we invert $G\dilat_p$ \emph{modulo} $X^p-1$
  using dense arithmetic and we multiply the result by $F\dilat_p$.  Then to get $[(Q\dilat)']_p$, we compute the derivatives of
  $F\dilat$ and $G\dilat$ and perform two more multiplications and one addition of dense polynomials, according to
  Equation~\eqref{eq:quoprimemodp}. All dense polynomial operations cost $\gOt(ps\log q)$ bit operations while the first step cost
  is given by Lemma~\ref{lem:Aalphap}. This concludes the proof since derivative cost is negligible.
\end{proof}

Huang's algorithm recovers monomials of $Q$ from $Q_p$ and $[Q']_p$. In our approach, we compute $Q\dilat_p$ and
$[(Q\dilat)']_p$ instead, and thus recover monomials of $Q\dilat$ instead of $Q$. Yet, if $c X^e$ is a monomial of $Q\dilat$,
the corresponding monomial in $Q$ is $c\alpha^{-e}X^e$ and it can be computed in $\gOt(\log(e)s\log q) = \gO(\log D s\log q)$ bit
operations.

\subsection{Case of large characteristic}

We first consider the case where the characteristic of $\cfq$ is larger than the degree $D$ of the input polynomials.  We begin
with the main ingredient of Huang's algorithm to further adapt it to our needs. Recall that the idea is to recover $Q$ from $Q_p$
and $[Q']_p$.

\begin{definition}
  Let $Q\in\cfq[X]$ and $p$ a prime number. Then $\algo{DLift}(Q_p, [Q']_p)$ is the polynomial $\hat Q = \sum_e cX^e$ where
  the sum ranges over all the integers $e$ such that $Q_p$ contains the monomial $cX^{e\bmod p}$ and $[Q']_p$ contains the
  monomial $ce X^{(e-1)\bmod p}$.
\end{definition}

Clearly, if one knows both $Q_p$ and $[Q']_p$, $\algo{DLift}(Q_p, [Q']_p)$ can be computed in $\gOt(p\log q)$ bit operations.
Next lemma revamps the core of Huang's result~\cite{Huang2019}. Similar results are used in several interpolation
algorithms~\cite{GiRo2011,ArGiRo13,HuangGao2020}.

\begin{lem}\label{lem:Huang}
  Let $Q\in\cfq[X]$ of degree at most $D$ and sparsity at most $T$. Let $p_1$, \dots, $p_k$ be randomly chosen among the first $N$
  prime numbers, where $N=\max(1,\lceil 12(T-1)\log D\rceil)$. Let $i$ that maximizes $\#Q_{p_i}$ and
  $\hat Q = \algo{DLift}(Q_{p_i}, (Q')_{p_i})$. Then with probability at least $1-2^{-k}$, $\#(Q-\hat Q) \le T/2$.
\end{lem}

The main idea in Huang's algorithm is to use this lemma $\log(T)$ times to recover all the coefficients of $Q$ with probability at
least $(1-2^{-k})^{\log T}$.
To extend the algorithm to our case, we need to compute $Q_p$ and $[Q']_p$ by choosing $\alpha$ in a suitable extension of $\cfq$
and compute $F(\alpha X)/G(\alpha X)\bmod X^p-1$ as explained in Section~\ref{sec:reducedquo}. Next corollary establishes the size
of that extension.

\begin{cor}\label{cor:Huang}
  Let $G\in\cfq[X]$ of degree at most $D$, and $\alpha$ be a random element of $\cfqs$ where
  $s = \lceil\log_q(\frac{965}{\epsilon}D^4)\rceil$. Then with probability at least $1-\epsilon$, $G(\alpha X)$ is coprime with
  $X^p-1$ for each of the $N$ first prime numbers, where $N$ is defined as in Lemma~\ref{lem:Huang}.
\end{cor}

\begin{proof}
  The polynomial $G(\alpha X)$ is coprime with all these polynomials if and only if it is coprime with their product. The degree
  of their product is the sum of the $N$ first prime numbers, which is bounded by $N^2\ln N$ (for $N > 3$). By
  Lemma~\ref{lem:coprime}, the probability that $G(\alpha X)$ be coprime with this product is at least $1-DN^2\ln N/q^s$ if
  $\alpha$ is chosen at random in $\cfqs$. Since $s\ge \log_q(\frac{144}{\epsilon}D^4)$, $q^s\ge \frac{144}{\epsilon}
  D^4$. Furthermore, since $T\le D$, $N\le 12D\log D$. This implies
\[\frac{DN^2\ln N}{q^s} \le \frac{144D^3\log^2(D)\ln(12D\log D)}{\frac{965}{\epsilon} D^4} \le \epsilon\]
since $144\log^2(D)\ln(12D\log D)\leq 965 D$ for all $D\ge 1$.
\end{proof}

By hypothesis on $\cfq$, we have $q\ge D$. This implies that $s=\gO(1)$ in Corollary~\ref{cor:Huang} as long as
$\frac{1}{\epsilon}$ remains polynomial in $D$. We now provide an algorithm for the exact division of sparse polynomials over large
finite fields, given a bound on the sparsity of the quotient.

\begin{algorithm}
\caption{\textsc{SparseDivLargeCharacteristic}}
\label{algo:sparsedivlargechar}
\begin{algorithmic}[1]
  \Input $F$, $G\in\cfq[X]$ such that $G$ divides $F$; a bound $T$ on $\#F$, $\#G$ and $\#(F/G)$; $0<\epsilon<1$
  \Output $F/G\in\cfq[X]$ with probability $\ge1-\epsilon$ 

\State Let $k = \lceil\log(\frac{2}\epsilon \log T)\rceil$ 
        and $N = \max(1,\lceil 12(T-1)\log D\rceil)$ 
\State Compute the set $\mathcal P$ of the first $N$ prime numbers \label{step:prime}
\State Compute an extension field $\cfqs$ where $s = \lceil\log_q(\frac{1930}{\epsilon}D^4)\rceil$ \label{step:ext}
\State Choose $\alpha$ at random in $\cfqs$
\State Compute $F\dilat = F(\alpha X)$ and $G\dilat = G(\alpha X)$ and set $\hat Q\dilat = 0$
\Loop \ $\lceil\log T\rceil$ times
    \State Choose $p_1$, \dots, $p_k$ at random in $\mathcal P$
    \For{each $p_i$}
        \State Compute $F\dilat_{p_i}$, $G\dilat_{p_i}$ \Comment Lemma~\ref{lem:Aalphap}
        \If{$G\dilat_{p_i}$ is not coprime with $X^{p_i}-1$} \textbf{return} \emph{failure} \EndIf
	\State Compute $Q\dilat_{p_i}$ and $[(Q\dilat)']_{p_i}$ \Comment Lemma~\ref{lem:Qalphap}
    \EndFor
    \State Let $p\in\{p_1,\dotsc,p_k\}$ such that $\#Q\dilat_p$ is maximal
    \State Add $\algo{DLift}(Q\dilat_p-\hat Q\dilat_p, [(Q\dilat)']_p- [(\hat Q\dilat)']_p)$ to $\hat Q\dilat$
\EndLoop
\State Return $\hat Q\dilat(\alpha^{-1}X)$
\end{algorithmic}
\end{algorithm}

\begin{thm}\label{thm:sparsedivlargechar}
Algorithm~\ref{algo:sparsedivlargechar} is correct. It uses $\gOteps(T\log D\log q)$ bit operations.
\end{thm}

\begin{proof}
  For the algorithm to succeed, $G(\alpha X)$ must be coprime with all the polynomials $X^p-1$ used in the loop. By
  Corollary~\ref{cor:Huang}, it is coprime with all the polynomials $X^p-1$ for $p\in\mathcal P$ with probability at least
  $1-\frac{\epsilon}{2}$. Next, the algorithm succeeds if at each iteration, the number of monomials of $Q\dilat-\hat Q\dilat$
  is halved. According to Lemma~\ref{lem:Huang}, this probability is at least
  $(1-2^{-k})^{\log T}\ge 1-2^{-k}\log T \ge 1-\frac{\epsilon}{2}$. Therefore, the overall probability of success is at least
  $1-\epsilon$.

  Let us first note that $s\log q = \gOeps(\log q)$ since $D\le q$. Step~\ref{step:prime} takes $\gOt(N) = \gOt(T\log D)$ bit
  operations while step~\ref{step:ext} takes $\gOt(s^3\log q) = \gOteps(\log q)$ bit operations. Computing the polynomials
  $F\dilat$ and $G\dilat$ costs $\gOt(T\log Ds\log q)=\gOteps(T\log D\log q)$ bit operations by Lemma~\ref{lem:Aalphap}. In the
  loop, as $p_i = \gO(N\log N)$, computing $F\dilat_{p_i}$ and $G\dilat_{p_i}$ costs $\gOteps(T(\log D+\log q))$ bit
  operations. Then operations on polynomials of degree at most $p$ take $\gOt(ps\log q)$ bit operations, that is
  $\gOteps(T\log D\log q)$ bit operations. Since this must be done for $\gOteps(\log T)$ primes in total, the overall cost of the
  algorithm is $\gOteps(T\log D\log q)$ bit operations.
\end{proof}

\subsection{Case of small characteristic}\label{ssec:smallchar}

When the field $\cfq$ has a characteristic smaller than the degree, Huang's technique is no more possible, and the best
alternative is to use the algorithm of Arnold, Giesbrecht and Roche~\cite{ArGiRo14}.
As mentioned before, their sparse interpolation algorithm computes $Q\dilat_p = Q(\alpha X)\bmod X^p-1$ for several values of
$\alpha$ and $p$ and they use the Chinese Remainder Theorem to recover the coefficients of $Q$.  This is the last part of~\cite[procedure BuildApproximation]{ArGiRo14}, and we denote it by \algo{CrtLift}.
Next lemma summarizes
their approach. It is the combination of~\cite[Lemma~3.1]{ArGiRo14} for the value of $\lambda$,~\cite[Corollary 3.2]{ArGiRo14} for
$\gamma$ and~\cite[Lemma~4.1]{ArGiRo14} for $m$ and $s$.

\begin{lem}[\cite{ArGiRo14}]\label{lem:AGR}
  Let $Q\in\cfq[X]$  a sparse polynomial of degree $D$ and sparsity $T$. Let $0<\mu<1$,
  $\lambda =\max(21,\lceil\frac{40}{3}(T-1)\ln D\rceil)$, $\gamma=\lceil\max(8\log_\lambda D,8\ln\frac{2}{\mu})\rceil$,
  $m=\lceil\log\frac{1}{\mu}+2\log(T(1+\frac{1}{2}\lceil\log_\lambda D\rceil))\rceil$ and $s\ge\log_q(2D+1)$.
  Let $\hat Q = \algo{CrtLift}((Q_{p_i}\dilatj)_{ij})$ where
  $Q_{p_i}\dilatj = Q(\alpha_j X)\bmod X^{p_i}-1$ 
  for random primes $p_1$, \dots, $p_\gamma$ in $]\lambda,2\lambda[$ and 
  random nonzero elements $\alpha_1$, \dots, $\alpha_m$ of $\cfqs$.  Then with
  probability at least $1-\mu$, $\#(Q-\hat Q) \le T/2$.
\end{lem}

In order to use such a lemma for sparse polynomial division, we need to compute $Q_{p_i}\dilatj$ for the $\gamma$ primes $p_i$
and the $m$ points $\alpha_j$.  As explained in Section~\ref{sec:reducedquo}, $Q_{p_i}\dilatj$ can be efficiently computed as
soon as $G(\alpha_jX)$ and $X^{p_i}-1$ are coprime.  To ensure, with good probability, that the coprimality property holds for every
$p_i$ and $\alpha_j$, we choose the $\alpha_j$'s in a somewhat larger extension of $\cfq$. That is, we need to increase the bound on
$s$ given in Lemma~\ref{lem:AGR} according to Lemma~\ref{lem:coprime}. 

\begin{lem}\label{lem:sizeExtAGR} 
  Let $G\in\cfq[X]$ of degree-$D$, $0<\mu<1$, $\lambda$, $\gamma$, $m$ three integers and $p_1$, \dots, $p_\gamma$ be prime numbers in
  $]\lambda,2\lambda[$. Let $s=\lceil\log_q(2\frac{\lambda\gamma}{\mu}mD)\rceil$ and $\alpha_1$, \dots, $\alpha_m$ be random
  elements of $\cfqs$. Then with probability at least $1-\mu$, $G(\alpha_j)$ and $X^{p_i}-1$ are coprime for all pairs $(i,j)$,
  $1\le i\le \gamma$ and $1\le j\le m$.
\end{lem}

\begin{proof}
	Let $\Pi=\prod_{i=1}^\gamma X^{p_i}-1$. Its degree is $\sum_{i=1}^\gamma p_i \le 2\lambda\gamma$.
	For any $\alpha_j$, $G(\alpha_jX)$ is coprime with all the polynomials $X^{p_i}-1$ if and only if $G(\alpha_jX)$ is coprime with $\Pi$.
	Since $\alpha_j$ is randomly chosen in $\cfqs$ then by Lemma~\ref{lem:coprime}, the probability that $G(\alpha_jX)$ and $\Pi$ are not coprime is at most $(2\lambda\gamma D)/q^s\leq \mu/m$ by definition of $s$.
	Therefore the probability that there is at least one $\alpha_j$ such that $G(\alpha_jX)$ and $\Pi$ are not coprime is at most $\mu$. 
\end{proof}

Using Lemmas~\ref{lem:AGR} and~\ref{lem:sizeExtAGR}, we can adapt 
the sparse interpolation algorithm from~\cite{ArGiRo14} to the exact quotient of two sparse polynomials.
It requires a larger extension than the original algorithm, but the growth is negligible by Lemma~\ref{lem:AGR}.

\begin{algorithm}
\caption{\textsc{SparseDivSmallCharacteristic}}
\label{algo:sparsedivsmallchar}
\begin{algorithmic}[1]
  \Input $F$, $G\in\cfq[X]$ such that $G$ divides $F$; a bound $T$ on $\#F$, $\#G$ and $\#(F/G)$; $0<\epsilon<1$
  \Output $F/G\in\cfq[X]$ with probability $\ge1-\epsilon$ 

\State Let $\mu = \frac{\epsilon}{2\lceil\log T\rceil}$ and set $\lambda$, $\gamma$ and $m$ as in Lemma~\ref{lem:AGR}
\State Compute the set $\mathcal P$ of the prime numbers in $]\lambda, 2\lambda[$
\State Compute an extension field $\cfqs$ where $s = \lceil\log_q(2\frac{\lambda\gamma}{\mu} mD)\rceil$
\State Set $\hat Q= 0$
\Loop \ $\lceil\log T\rceil$ times
    \State Choose $p_1$, \dots, $p_\gamma$ at random in $\mathcal P$ 
    \State Choose $\alpha_1$, \dots, $\alpha_m$ at random in $\cfqs$
    \For{each pair $(p_i,\alpha_j)$}
        \State Compute $F\dilatj_{p_i}$, $G\dilatj_{p_i}$ \Comment Lemma~\ref{lem:Aalphap}
        \If{$G\dilatj_{p_i}$ is not coprime with $X^{p_i}-1$} \textbf{return} \emph{failure} \EndIf
        \State Compute $Q\dilat_{p_i}$ \Comment Lemma~\ref{lem:Qalphap}
    \EndFor
    \State Add $\algo{CrtLift}((Q\dilatj_{p_i}-\hat Q\dilatj_{p_i})_{i,j})$ to $\hat Q$
\EndLoop
\State Return $\hat Q$
\end{algorithmic}
\end{algorithm}

\begin{thm}\label{thm:sparsedivsmallchar}
Algorithm~\ref{algo:sparsedivsmallchar} is correct. 
It uses $\gOteps(T\log^2D(\log D+\log q))$ bit operations where $D = \deg(F)$.
\end{thm}

\begin{proof}
  The algorithm is a modification of~\cite[Procedure~MajorityVoteSparseInterpolate]{ArGiRo14}. 
  The algorithm succeeds if at each iteration, every $G_{p_i}\dilatj$ is coprime with $X^{p_i}-1$ and if \algo{CrtLift}
  succeeds in recovering at least half of the terms of $Q-\hat Q$. Both conditions hold with probability $1-\mu$ at each
  step. The global probability of success is thus at least
  $(1-\mu)^{2\lceil\log T\rceil}\ge 1-2\mu\lceil\log T\rceil \ge 1-\epsilon$.

  As in the original algorithm, the cost is dominated by the computation of all the $Q_{p_i}\dilatj$. There are
  $\gamma m\lceil \log T\rceil$ such polynomials to compute. Since $p_i < 2\lambda$ and $\alpha_j\in\cfqs$, each computation costs
  $\gOt((T\log D+\lambda)s\log q)$ bit operations according to
  Lemma~\ref{lem:Qalphap}. 
  As $\lambda=\gO(T\log D)$, $\gamma=\gOteps(\log D)$, $m$ is logarithmic and $s=\gOteps(1+\log_q D)$, the algorithm requires
  $\gOteps(T\log^2D(\log D+\log q))$ bit operations.
\end{proof}

\subsection{Output sensitive algorithm}

Both interpolation algorithms presented in the previous sections require a bound on the sparsity of the quotient.  To overcome
this difficulty, we use the same strategy as for sparse polynomial multiplication~\cite{giorgi2020:sparsemul}. The idea is to
guess the sparsity bound as we go using a fast verification procedure. For verifying exact quotient, we can directly reuse
Algorithm \verif from~\cite{giorgi2020:sparsemul} that verifies if $F=GQ$, with an error probability at most $\epsilon$ if
$F\neq GQ$.  It requires 
$\gOteps(T(\log D+\log q))$ bit operations over $\cfq$ and $\gOteps(T(\log D+\log C))$ bit operations over $\gz$ where $C$ is a
bound on the heights of $F$, $G$ and $Q$.

\begin{algorithm}
\caption{\textsc{SparseExactDivision}}
\label{algo:exactdiv}
\begin{algorithmic}[1]
	\Input $F$, $G\in\cfq[X]$ such that $G$ divides $F$; $0<\epsilon<1$
\Output $F/G\in\cfq[X]$ with probability at least $1-\frac{\epsilon}{2}$

	\State $t\gets 1$
	\Repeat
    		\State $t\gets2t$
		\State \label{exactdiv:tentative}
                    Compute a tentative quotient $\hat Q$ using Algorithm~\ref{algo:sparsedivlargechar} or~\ref{algo:sparsedivsmallchar},\Statex \hspace*{\algorithmicindent} with sparsity bound $t$ and probability $\frac{\epsilon}{2}$ 
	\Until{$\verif(F,G,\hat Q,\frac{\epsilon}{2t})$}
	\State \Return $\hat Q$
\end{algorithmic}
\end{algorithm}

\begin{thm}\label{thm:exactdiv} 
  Let $F$, $G\in\cfq[X]$ such that $G$ divides $F$, $0<\epsilon<1$, $D=\deg(F)$ and $T= \max(\#F,\#G, \#(F/G))$.  With probability
  at least $1-\epsilon$, Algorithm~\ref{algo:exactdiv} returns $F/G$ in $\gOteps(T\log D\log q)$ bit operations if 
  $\operatorname{char}(\cfq)>D$  or~$\gOteps(T\log^2D(\log D+\log q))$ otherwise.
\end{thm} 
 
\begin{proof}
  The probability $1-\epsilon$ concerns both the correctness and the complexity of the algorithm. More precisely, we prove that
  the algorithm is correct with probability $\ge 1-\frac{\epsilon}{2}$ and that it performs the claimed number of bit operations
  with probability $\ge 1-\frac{\epsilon}{2}$.

  The algorithm is incorrect when $F\neq G\hat Q$. This happens if at some iteration, Algorithm~\ref{algo:sparsedivlargechar}
  or~\ref{algo:sparsedivsmallchar} returns an incorrect quotient but the verification algorithm fails to detect it.  
  In other words, for the algorithm to return a correct answer, all the verifications must succeed.
  This happens with probability at least
  $1-\sum_t \frac{\epsilon}{2t}\ge 1-\frac{\epsilon}{2}$ since the sum ranges over powers of two.\footnote{The error
    probability analysis of~\cite[Algorithm~2]{giorgi2020:sparsemul} is flawed and should be replaced by this new one.}

  For the complexity we first need to bound the number of iterations. Since the values of $t$ are powers of two, the first value
  $\ge \#(F/G)$ is at most $2\#(F/G)$. If $t$ attains this value, Algorithm~\ref{thm:sparsedivlargechar}
  or~\ref{thm:sparsedivsmallchar} correctly computes $F/G$ with probability at least $1-\frac{\epsilon}{2}$ according to
  Theorems~\ref{algo:sparsedivlargechar} and~\ref{algo:sparsedivsmallchar}. That is, with probability at least
  $1-\frac{\epsilon}{2}$, $t$ is bounded by $2T$ and the number of iterations is $\gO(\log T)$.
  Depending on the characteristic, using Theorems~\ref{thm:sparsedivlargechar} or~\ref{thm:sparsedivsmallchar} and the complexity
  of~\cite[Algorithm~$\verif$]{giorgi2020:sparsemul}, we obtain the claimed complexity with probability at
  least $1-\frac{\epsilon}{2}$.
\end{proof}

\subsection{Algorithm over the integers}
For polynomials over $\gz[X]$, we cannot directly use Algorithm~\ref{algo:exactdiv} with the variant of Huang's algorithm over
$\gz$~\cite{Huang2019}. Indeed, the coefficients arising during the computation may be dramatically larger than the inputs and the
output. This is mostly due to the inversion of $G$ \emph{modulo} $X^p-1$.  Instead we use the standard technique that maps the
computation over some large enough prime finite field.  As we cannot tightly predict the size of the coefficients, we can again
reuse our guess-and-check approach with several prime numbers of growing size to discover it as we go. In order to use the fastest
algorithm (Algorithm~\ref{algo:sparsedivlargechar}) we consider prime finite fields $\cfq$ such that $q$ is larger than the input
degree.

We first define a bound on the coefficient of the quotient of two sparse polynomial over $\gz[X]$ as the classic Mignotte's bound
\cite{MCAlgebra} on dense polynomial is too loose and it has no equivalent in the sparse case.

\begin{lem}\label{lem:sdivheighbound}
  Let $F$, $G$, $Q\in\gz[X]$ be three sparse polynomials such that $F=QG$ and $T=\#Q$ is the number of nonzero coefficient of
  $Q$. Then
  \[
    \height{Q}  \leq (\height{G}+1)^{\ceil{\frac{T-1}{2}}}\height{F}. 
  \]
\end{lem}
\begin{proof}
  Write $Q=\sum_{i=1}^T q_i X^{e_i}$ with $e_1>e_2>\dots>e_T$. We use
  induction on the remainder and quotient sequence in the Euclidean division algorithm.  Let $Q_j= \sum_{i=1}^jq_iX^{e_i}$ and $R_j=F-Q_jG$
  be the elements of that sequence, starting with $R_0=F$ and $Q_0=0$. The integer coefficients of $Q$ are defined as
  $q_{i}= \mathsf{LC}(R_{i-1})/ \mathsf{LC}(G)$ where $\mathsf{LC}$ denote the leading coefficient.  We know from the algorithm
  that $R_{i}= R_{i-1} -q_{i}X^{e_{i}}G$ and $Q_{i}=Q_{i-1}+q_{i}X^{e_{i}}$. Since $R_0=F$ and
  \[
    \height{R_{i}} \leq \height{R_{i-1}} + |q_{i}| \times \height{G} \leq \height{R_{i-1}}
    (1+\height{G})
  \]
  we have $\height{R_{i}} \leq \height{F}(1+\height{G})^{i}$. Since the reciprocal $Q^\star$ of
  $Q$ is defined by the quotient $F^\star/G^\star$, we also get $\height{R_i}\le\height F(1+\height G)^{T-i}$. 
  Therefore, 
  \[\height{Q}= \max_i(|q_i|) \leq \max_i\height{R_{i}}\leq \height{F}(1+\height{G})^{\lceil\frac{T-1}{2}\rceil}.\qedhere\] 
\end{proof}

\begin{algorithm}
\caption{\textsc{SparseExactDivisionOverZ}}
\label{algo:divoverZ}
\begin{algorithmic}[1]
	\Input $F$, $G\in\gz[X]$ such that $G$ divides $F$; $0<\epsilon<1$
\Output $F/G$ with probability at least $1-\frac{\epsilon}{2}$

	\State Let $n = \deg(F)$, $i=2$
	\Repeat
		\State $i\gets 2i$
		\State Choose $q$ at random in $]n,2n[$, prime with prob. $\ge1-\frac{\epsilon}{2i}$
		\State Compute the reductions $F_q = F\bmod q$ and $G_q = G\bmod q$
		\State Compute $\hat Q = \textsc{SparseExactDivision}(F_q,G_q,\frac{\epsilon}{2i})$
    		\State $n\gets n^2$
	\Until{$\verif(F,G,\hat Q,\frac{\epsilon}{i})$}
	\State \Return $\hat Q$
\end{algorithmic}
\end{algorithm}

\begin{thm}\label{thm:divoverZ}
  Let $F$, $G$ be two sparse polynomials in $\gz[X]$ such that $G$ divides $F$, $0<\epsilon<1$, $D=\deg(F)$,
  $C=\max(\height F+\height G)$ and $T=\max(\#F,\#G,\#(F/G))$.  With probability at least $1-\epsilon$,
  Algorithm~\ref{algo:divoverZ} returns $F/G$ in $\gOteps(T(\log C+\log^2 D)+\log^3D)$ bit operations if $D>2\height{F/G}$ or
  $\gOteps(T(\log C+\log D\log \height{F/G})+\log^3\height{F/G})$ bit operations otherwise.
\end{thm}

\begin{proof}
  The proof goes along the same lines as for Theorem~\ref{thm:exactdiv}. 
  With the same arguments, the probability that the algorithm returns an incorrect quotient is at most
  $\frac{\epsilon}{2}$.

  In Step~4, we can use a Miller-Rabin based algorithm to compute a number $q$ that is prime with probability at least
	$1-\frac{\epsilon}{2i}$ in $\gOt_{{\epsilon}/{i}}(\log^3q)$ bit operations.  Step~5 performs $\gO(T)$ divisions in $\gOt(T\log C)$ bit operations. Thus
  by Theorem~\ref{thm:exactdiv}, 
  an iteration of the loop 
	correctly computes $F_q/G_q$ in $\cfq[X]$
	in $\gOt_{{\epsilon}/{i}}(T\log D\log q+T\log C +\log^3q)$ bit operations
	with probability at least $1-\frac{\epsilon}{i}$. 
  Let $Q = F/G$. As soon as $n>2\height{Q}$, $F_q/G_q$ is actually $Q$.
	Therefore, the algorithm stops with $n<4\height{Q}^2$ with probability
	$\ge 1-\frac{\epsilon}{i}$.  If $D>2\height{Q}$, $q$ satisfies $2\height Q<q<2D$ at the first
  iteration. 
  Hence, the algorithm correctly computes $Q$ in one iteration with probability $\ge1-\frac{\epsilon}{2}$. Its bit
  complexity is then $\gOt_{\epsilon}(T\log^2 D+T\log
  C+\log^3D)$. 

  Otherwise, at most $j=\ceil{\log\log\height Q}$ iterations are needed to get $2\height Q<n<q<4\height Q^2$. 
	Thus $i$ is bounded by $2^{j+1}$
	and $Q$ is correctly
	computed in $\gOt_{\epsilon/{2^j}}(T\log D\log\height Q+T\log C+\log^3\height Q)$ bit operations with probability at least
	$1-\sum_{k\geq2}\frac{\epsilon}{2^k}\geq1-\frac{\epsilon}{2}$. By Lemma~\ref{lem:sdivheighbound},
	$\height Q\leq (\height G+1)^{\ceil{\frac{T-1}{2}}}\height F$, whence $\log\frac{2^j}{\epsilon}=\gO(\log T+\log\log C+\log\frac{1}{\epsilon})$
	and $\gOt_{\epsilon/2^j}(\cdot)$ is $\gOteps(\cdot)$.
	
  In both cases, 
  the cost of the loop body 
  is as stated with probability at
  least $1-\frac{\epsilon}{2}$.  Since the verification with probability of success at least $1-\frac{\epsilon}{2i}$ requires
  $\gOt_{\epsilon/2i}(T(\log D +\log C+ \log \height Q))$ bit operations and the maximal value of $i$ is expected to be $\gO(\log\height Q)$,
  its cost is negligible compared to the loop body. Thus the algorithm works as stated with probability at least $1-\epsilon$.
\end{proof}

\section{Divisibility testing}\label{sec:divtesting}

Given two sparse polynomials $F$, $G\in\K[X]$, we want to check whether $G$ divides $F$ in polynomial time. If $\deg(F) = m+n-1$,
$\deg(G) = m$ and $\#F$, $\#G\le T$, then the input size is $\gO(T\log(m+n))$ and the divisibility check must cost
$(T\log(m+n))^{\gO(1)}$. We first remind the only known positive results.

\begin{prop} \label{prop:basecases} Let $F$, $G\in\K[X]$ of degrees $m+n-1$ and $m$ respectively, and sparsity at most $T$. If either $m$ or $n$ is polynomial in $T\log(m+n)$, one can check whether $G$ divides $F$ in polynomial time. 
\end{prop}

\begin{proof} 
  Let $F = GQ+R$ be the Euclidean division of $F$ by $G$.  When $n$ is polynomially bounded, the Euclidean division algorithm runs
  in polynomial time and the verification is trivial.  When $m$ is polynomially bounded, the degree of the remainder is polynomial
  and it can be computed in polynomial time without computing $Q$. Indeed, it suffices to compute $X^e\bmod G$ for each exponent
  $e\in\supp(F)$ in polynomial time by fast exponentiation.
\end{proof}

The rest of the section can be seen as a generalization of the proposition.  As long as one has a polynomial bound on the size of
the quotient, the divisibility test is polynomial by either computing $F-GQ$ or asserting that $F = GQ$.  We begin with a very
simple remark, that we shall use repeatedly.

\begin{rem}\label{rem:exactdivreversal}
  Let $F$, $G\in\K[X]$, and $F^\star = X^{\deg F}F(1/X)$ and $G^\star = X^{\deg G} G(1/X)$ be their reciprocal polynomials. Then
  $G$ divides $F$ if and only if $G^\star$ divides $F^\star$. In this case, the quotient $Q^\star =F^\star\bquo G^\star$ is the
  reciprocal of the quotient $Q = F\bquo G$.
\end{rem}

\begin{proof}
  From the definition, $(AB)^\star = A^\star B^\star$ for $A$, $B\in\K[X]$. Therefore, if there exists $Q$ such
  that $F = GQ$, then $F^\star = G^\star Q^\star$. The converse follows since the reciprocal is involutive.
\end{proof}

We note that the equality $(F\bquo G)^\star = F^\star\bquo G^\star$ is not true in general when $G$ does not divide $F$. Let $F = GQ+R$ with $\deg(R)<\deg(G)$. Then $F^\star = X^{\deg(F)}G(1/X)Q(1/X) + X^{\deg(F)}R(1/X) = G^\star Q^\star + X^{\deg(F)}R(1/X)$. 
Therefore $Q^\star = F^\star\bquo G^\star$ if and only if $R = X^{\deg(F)}R(1/X)$.

\subsection{Bounding the sparsity of the quotient}

In this section, we provide a bound on the sparsity of the quotient $Q = F\bquo G$, depending on the \emph{gap} between the
highest and second highest exponents in $G$. We make use of the following estimation.

\begin{lem}\label{lem:sparsitybound}
  Let $G\in\K[X]$ of degree $m$ and sparsity $\#G$, such that $G=1 - X^k G_1$ with $G_1\in\K[X]$ of degree $m-k$. Then for all
  $t\ge 0$, $1/G\bmod X^{kt}$ has at most $\frac{1}{(t-1)!}(\#G+t-2)^{t-1}$ nonzero monomials.
\end{lem}

\begin{proof}
  Since $G(0) = 1$, it is invertible in the ring of power series. Let $\phi = \sum_{i\ge 0} X^{ki}G_1^i\in\K[[X]]$ be its inverse. 
  As soon as $i\ge t$, $X^{ki}\bmod X^{kt}= 0$. Therefore
  \[ \phi\bmod X^{kt} = \sum_{i=0}^{t-1} X^{ki}G_1^i\bmod X^{kt}.\]
  Note that the support of $G^{t-1}$ is a subset of $S=\{\sum_{j=1}^{t-1} e_j:e_j\in\supp(G)\}$ which has size at most
  $\binom{\#G+t-2}{t-1}\le \frac{1}{(t-1)!}(\#G+t-2)^{t-1}$.  Using the expansion
  $G^{t-1} = \sum_{i=0}^{t-1} \binom{t-1}{i}(-1)^iX^{ki}G_1^i$, one can  see by identification that
  $\supp(\phi\bmod X^{kt}) \subseteq S$, whence the result.  
\end{proof}

\begin{cor}\label{cor:sparsitybound}
  Let $F$ and $G\in\K[X]$ of respective degrees $m+n-1$ and $m$, and respective sparsities $\#F$ and $\#G$. If $G=X^m - G_0$ with
  $G_0\in\K[X]$ of degree $m-k$ then the quotient $Q=F\bquo G$ has at most $\frac{1}{(\lceil n/k\rceil-1)!}\#F(\#G+\lceil n/k\rceil-2)^{\lceil n/k\rceil-1}$
  nonzero monomials.
\end{cor}

\begin{proof}
  Let $F = GQ+R$ with $\deg(R)<m$. It is classical that the reciprocal $Q^\star$ of $Q$ equals
  $F^\star/G^\star\bmod X^n$~\cite{MCAlgebra}. We can apply Lemma~\ref{lem:sparsitybound} to $G^\star$ since
  $G^\star = 1-X^{k}G_0^\star$. Hence, $1/G^\star\bmod X^n$ has at most
  $\frac{1}{(\lceil n/k\rceil -1)!}(\#G-\lceil n/k\rceil -2)^{\lceil n/k\rceil-1}$ nonzero monomials, using $t = \lceil n/k\rceil$ and noting that
  $n\le kt$. This implies that the sparsity of $Q^\star$, that is the sparsity of $Q$, is at most
  $\frac{1}{(\lceil n/k\rceil-1)!}\#F(\#G+\lceil n/k\rceil-2)^{\lceil n/k\rceil-1}$.
\end{proof}

\begin{cor}\label{cor:sparsityboundbis}
  Let $F$, $G\in\K[X]$ of respective degrees $m+n-1$ and $m$. If $G = 1 - X^k G_1$ and $G$ divides $F$, then the quotient $Q=F\bquo G$ has at
  most 
$\frac{1}{(\lceil n/k\rceil-1)!}\#F(\#G+\lceil n/k\rceil-2)^{\lceil n/k\rceil-1}$
  nonzero monomials.
\end{cor}

\begin{proof}
  We apply Corollary~\ref{cor:sparsitybound} to $F^\star$ and $G^\star$. Indeed, $G^\star =X^m-G_0$ for some $G_0$ of degree $m-k$
  and since $G$ divides $F$, we have $F\bquo G = (F^\star\bquo G^\star)^\star$.
\end{proof}

Next example shows that the bound does not hold anymore if $G$ does not divide $F$.

\begin{example}
Let $F = X^{m+n-1}-1$ and $G = X^m-X^{m-1}+1$. Then $F\bquo G = \sum_{i=0}^{n-1} X^i$ is as dense as possible.
\end{example}

If $F = GQ+R$ with some nonzero $R$ then obviously $F-R = GQ$, that is $G$ divides $F-R$. This implies that if $R$ has few nonzero
monomials, then $Q$ as well since $F-R$ is a sparse polynomial. Conversely, if $Q$ has few nonzero monomials, $R = F-GQ$ as well. As a
result, we observe that the sparsities of the quotient and the remainder in the Euclidean division of $F$ by $G = 1+X^k G_1$ are
polynomially related.

\subsection{Algorithmic results}

Let $F$, $G\in\K[X]$ of respective degrees $m+n-1$ and $m$, with $n = \gO(m)$. Results of the previous section show that if
$G = X^m - G_0$ with $\deg(G_0)\le m-k$ for some $k = \gO(m)$, the sparsity of the quotient $F\bquo G$ is polynomially bounded in
the input size. If $G = 1+X^kG_1$, the same holds for the quotient $F^\star\bquo G^\star$. In both cases, this implies that one
can check whether $G$ divides $F$ by a mere application of the Euclidean division algorithm.  Our aim is to extend this approach
to a larger family of divisors $G$ through a generalization of Lemma~\ref{lem:sparsitybound}. It is based on the following lemma.

\begin{lem}\label{lem:div}
Let $F$, $G$ and $C\in\K[X]$, $C\neq 0$. Then $G$ divides $F$ if and only if $G$ divides $FC$ and $C$ divides $FC/G$.
\end{lem}

\begin{proof}
If $G$ divides $F$, then $G$ clearly divides $FC$. Writing $F = GQ_1$, it is also clear that $C$ divides $FC/G = Q_1C$.
Conversely, if $G$ divides $FC$ and $C$ divides $FC/G$, we can write $FC/G = CQ_2$. Hence $F = GQ_2$ and $G$ divides $F$.
\end{proof}

The generalization of Lemma~\ref{lem:sparsitybound} is given by the following lemma.

\begin{lem} \label{lem:gensparsitybound} Let $G\in\K[X]$ of degree $m$ and sparsity $\#G$, such that $G = G_0 - X^k G_1$ with
	$\deg(G_0) < k$ and $G(0)\neq 0$. Then for all $t$, $G_0^t/G\bmod X^{tk}$ has at most $\frac{1}{(t-1)!}(\#G+t-2)^{t-1}$ nonzero
  monomials.
\end{lem}

\begin{proof}
  Expanding $G_0/G = 1/(1-X^kG_1G_0^{-1}) = \sum_{i\ge 0} X^{ki}G_1^iG_0^{-i}$, we get 
  $G_0^t/G = \sum_{i\ge 0} X^{ki}G_1^i G_0^{t-i-1}$ for all $t$. Since $X^{ki}\bmod X^{kt} = 0$ for $i\ge t$, 
  \[G_0^t/G\bmod X^{kt} = \sum_{i=0}^{t-1} X^{ki} G_1^i G_0^{t-1-i} \bmod X^{kt}.\] The support of $G_0^t/G\bmod X^{kt}$ is also a
  \sloppy
  subset of $S$ defined in the proof of Lemma~\ref{lem:sparsitybound} since
  $G^{t-1}= \sum_{i=0}^{t-1} \binom{t-1}{i} (-1)^iX^{ki} G_1^i G_0^{t-1-i}$.  Therefore, its sparsity is at most
  $ \frac{1}{(t-1)!}(\#G+t-2)^{t-1}$.
\end{proof}

\begin{thm} \label{thm:divisibility} Let $F$ and $G\in\K[X]$ be two sparse polynomials, of degrees $m+n-1$ and
  $m$ respectively, and sparsity at most $T$. One can check whether $G$ divides $F$ in polynomial time if $G = G_0 - X^k
  G_1$ where $k-\deg(G_0) = \Omega(n)$ and either $\deg(G_0)$ or
  $\deg(G_1)$ is bounded by a polynomial function of the input size.
\end{thm}

\begin{proof}
  We first note that we can first remove any power of $X$ that divides $F$ or $G$. If $X^a$ divides $F$ and $X^b$ divides $G$,
  then $G$ divides $F$ if and only if $b\le a$ and $G/X^b$ divides $F/X^a$. Therefore, we assume from now on that $G(0)$ and
  $F(0)$ are nonzero. This implies in particular that $G$ and $G_0$ are both invertible in the ring of power series over $\K$. We
  treat the case $\deg(G_0) = (T\log(m+n))^{\gO(1)}$. The second case is directly obtained by taking reciprocals.

  By Lemma~\ref{lem:div}, for any integer $t\ge 0$, $G$ divides $F$ if and only if $G$ divides $G_0^tF$ and $G_0^t$ divides
  $FG_0^t/G$. Our algorithm 
  checks these conditions for some $t$ such that $k-\ell\ge n/t$, where
  $\ell = \deg(G_0)$.

	By Lemma~\ref{lem:gensparsitybound}, $G_0^t/G\bmod X^{kt}$ has at most $\frac{1}{(t-1)!}(T+t-2)^{t-1}$ nonzero terms, whence
	$FG_0^t/G\bmod X^{kt}$ at most $\frac{1}{(t-1)!}T(T+t-2)^{t-1}$. Note that $t = \gO(1)$ since $k-\ell=\Omega(n)$, and 
  that $kt\ge n+\ell t$. Since $G_0^t/G\bmod X^{n+\ell t} = ((FG_0^t)^\star\bquo G^\star)^\star$, the sparsity of
  $(FG_0^t)^\star\bquo G^\star$ is at most $T^{\gO(1)}$. One can compute this quotient and check whether the
  remainder vanishes to test in polynomial time if $G$ divides $FG_0^t$. If the test fails, $G$ does not divide $F$. Otherwise,
  we have computed a polynomial $Q_0$ such that $FG_0^t = Q_0G$. It remains to check whether $G_0^t$ divides $Q_0$. Proposition~\ref{prop:basecases} 
  provides a polynomial-time algorithm for this since $\deg(G_0^t)$ is polynomially bounded.
\end{proof}

The previous proof extends to more general divisors. It only requires a polynomial
bound on the sparsity of $Q_0$ and a polynomial-time algorithm to test whether $G_0^t$ divides $Q_0$. The second step
can be a recursive call if $G_0^t$ satisfies the conditions in the theorem. This provides the following generalization of the theorem.

\begin{cor}\label{cor:divisibility}
  Let $F$ and $G\in\K[X]$ be two sparse polynomials, of degrees $m+n-1$ and $m$ respectively, and sparsity at most $T$.  One can
  check whether $G$ divides $F$ in polynomial time if $G = G_0 + X^k G_1 - X^\ell G_2$ with $G_0$, $G_1$, $G_2\in\K[X]$ such that
  $k-\deg(G_0)$ and $\ell-k-\deg(G_1)$ are both $\Omega(n)$ and $\deg(G_1) = (T\log(m+n))^{\gO(1)}$.
\end{cor}

\begin{proof}
We assume that $T\log(m+n) = n^{o(1)}$. Otherwise, one can use Proposition~\ref{prop:basecases}.
Using Lemma~\ref{lem:gensparsitybound}, $F(G_0+X^kG_1)^t/G\bmod X^{kt}$ has at most $T^{\gO(1)}$ nonzero monomials for $t = \gO(1)$. Therefore, as
previously, we can compute the quotient $(F(G_0+X^kG_1)^t)^\star\bquo G^\star$ for $t = \lceil n/(\ell-k-\deg(G_1))\rceil$, in
polynomial time. If the remainder is nonzero, $G$ does not divide $F$.
Otherwise, we have computed a polynomial $Q_{01}$ such that $F(G_0+X^k G_1)^t = Q_{01}G$. It remains to test whether
$H = (G_0+X^kG_1)^t$ divides $Q_{01}$. We show that the polynomial $H$ satisfies the conditions of Theorem~\ref{thm:divisibility}. Let
us write
\[H = \sum_{i=0}^t \binom{t}{i} X^{ki}G_1^i G_0^{t-i} = X^{kt}G_1^t + \sum_{i=0}^{t-1}\binom{t}{i} X^{ki}G_1^iG_0^{t-i} = X^{kt}G_1^t + H_0\]
where $H_0$ has degree at most $k(t-1)+\deg(G_1)(t-1)+\deg(G_0)$. Then $kt-\deg(H_0)\ge k-\deg(G_0)-(t-1)\deg(G_1) = \Omega(n)$
since $k-\deg(G_0) = \Omega(n)$ and $\deg(G_1) = (T\log (m+n))^{\gO(1)} = n^{o(1)}$. One can test whether $H$
divides $Q_{01}$ in polynomial-time using Theorem~\ref{thm:divisibility}.
\end{proof}

Theorem~\ref{thm:divisibility} and Corollary~\ref{cor:divisibility} cover cases were the quotient of the polynomials and the quotient of their reciprocals are both dense, as shown in the following example.

\begin{example}
Let $F=X^{2n-1}-X^n-X^{n-1}+X^2-X+1$ and $G=G_0 + X^{n-1}G_1$ where $G_0 = G_1 = 1-X$. Then $F\bquo G =\sum_{i=0}^{n-1}X^i$ and $F^\star\bquo G^\star = X^{n-1} + \sum_{i=0}^{n-3} X^i$.
\end{example}

\section*{Acknowledgments}
We are grateful to the reviewers for their insightful comments.

\newcommand{\Gathen}{\relax}\newcommand{\Hoeven}{\relax}

\end{document}